\documentclass[11pt]{article}

\usepackage[in]{fullpage}

\usepackage[twoside,letterpaper,margin=1in]{geometry}

\usepackage[utf8]{inputenc}
\usepackage{microtype}
\usepackage{graphicx,url,amsmath,amsfonts,amssymb,subfigure,bbm,bm,enumitem}
\usepackage{wrapfig}
\usepackage{cite}

\usepackage[dvips]{epsfig}
\usepackage{thm-restate}
\usepackage{float}
\usepackage{color}
\usepackage{dsfont}
\usepackage{algorithmicx}
\usepackage{algorithm}
\usepackage{algpseudocode}
\usepackage{tikz}
\usetikzlibrary{positioning,calc}
\usepackage{environ}

\makeatletter
\newsavebox{\box@tikzpicture}
\NewEnviron{tikzpicture*}[2][]{%
    \begin{lrbox}{\box@tikzpicture}%
        \begin{tikzpicture}[#1]
        \BODY
        \end{tikzpicture}%
    \end{lrbox}
    \pgfmathsetmacro\width@scale@picture{#2/\wd\box@tikzpicture}%
    \begin{tikzpicture}[#1,scale=\width@scale@picture]
    \BODY
    \end{tikzpicture}%
}%
\makeatother

\newcommand{\calV}{\mathcal{V}}
\newcommand{\calS}{\mathcal{S}}
\newcommand{\Sin}{\mathcal{S}_{\text{in}}}
\newcommand{\Sout}{\mathcal{S}_{\text{out}}}
\newcommand{\ins}{\text{in}}
\newcommand{\out}{\text{out}}

\usepackage{todonotes}

\definecolor{blueblack}{rgb}{0,0,.7}

\makeatletter
\newcommand*{\customnum}[1]{%
  \expandafter\@customnum\csname c@#1\endcsname%
}

\newcommand*{\@customnum}[1]{%
  $\ifcase#1\or4'
    \else\@ctrerr\fi$%
}
\AddEnumerateCounter{\customnum}{\@customnum}{53.13}
\makeatother

\usepackage{amsthm}
\theoremstyle{plain}
\newtheorem{theorem}{Theorem}[section]

\newtheorem{Definition}[theorem]{Definition}

\newtheorem{fact}[theorem]{Fact}
\newtheorem{claim}[theorem]{Claim}
\newtheorem{lemma}[theorem]{Lemma}

\newcommand{\maxrad}{r_q}

\newcommand{\cost}{cost}
\newcommand{\dist}{\text{dist}}
\newcommand{\price}{price}

\newcommand{\executeiffilenewer}[3]{%
\ifnum\pdfstrcmp{\pdffilemoddate{#1}}%
{\pdffilemoddate{#2}}>0%
{\immediate\write18{#3}}\fi%
}

\newcounter{sideremark}

\def\eg{{\it e.g.,}~}
\def\ie{{\it i.e.},~}
\newcommand{\eps}{\varepsilon}
\renewcommand{\epsilon}{\varepsilon}

\usepackage{authblk}

 \makeatletter
 \def\cramped
   {\parskip\@outerparskip\@topsep\parskip
   \@topsepadd2pt\itemsep0pt
 }
 \makeatother

\def\opt{{\text{OPT}}}

\newcommand{\R}{\mathbb{R}}
\DeclareMathOperator*{\argmax}{argmax}

\title{The Bane of Low-Dimensionality Clustering\thanks{      The project leading to this application has received funding from the European Union’s Horizon 2020 research and innovation programme under the Marie Sklodowska-Curie grant agreement No. 748094.  The work of A. de Mesmay is partially supported by the French ANR project ANR-16-CE40-0009-01 (GATO). The work of A. Roytman is partially supported by Thorup's Advanced Grant DFF-0602-02499B
from the Danish Council for Independent Research.
    }}

\author[1,3]{Vincent Cohen-Addad}
\author[2]{Arnaud de Mesmay}
\author[1]{Eva Rotenberg}
\author[1]{Alan Roytman}
\affil[1]{Department of Computer Science, University of Copenhagen, Denmark}
\affil[2]{Univ. Grenoble Alpes, CNRS, Grenoble INP, GIPSA-lab, 38000 Grenoble, France}
\affil[3]{Sorbonne Universit\'es, UPMC Univ Paris 06, CNRS, LIP6, Paris, France}

\date{}

\begin{document}
\maketitle

\begin{abstract}
In this paper, we give a conditional lower bound of $n^{\Omega(k)}$ on running time for 
the classic $k$-median and $k$-means clustering objectives (where $n$ is the size of the input), even in low-dimensional
Euclidean space of dimension four,
assuming the Exponential Time Hypothesis (ETH). We also consider $k$-median (and $k$-means) with penalties where each point need not be assigned to
a center, in which case it must pay a penalty, 
and extend our lower bound to at least three-dimensional Euclidean space.

This stands in stark contrast to many other geometric problems such as the
traveling salesman problem, or computing an independent set of unit
spheres. While these problems benefit from the so-called (limited) blessing
of dimensionality, as they can be solved in time $n^{O(k^{1-1/d})}$ or
$2^{n^{1-1/d}}$ in $d$ dimensions, our work shows that widely-used clustering objectives
have a lower bound of $n^{\Omega(k)}$, even in dimension four.

We complete the picture by considering the two-dimensional case: we show
that there is no algorithm that solves the penalized version in time less than $n^{o(\sqrt{k})}$, and provide a matching upper bound of $n^{O(\sqrt{k})}$.

The main tool we use to establish these lower bounds is the placement
of points on the moment curve, which takes its inspiration from
constructions of point sets yielding Delaunay complexes of high
complexity.

\end{abstract}
\setcounter{page}{0}
\thispagestyle{empty}
\newpage

\section{Introduction}
The fundamental $k$-median problem has led to several
important algorithmic
results since the beginning of its study in the 1970s~\cite{sahni1976p}.
It has consistently received  attention from both practitioners and
theoreticians, and there is now a vast literature on the problem in different
settings, such as streaming, fixed-parameter tractability (FPT), and beyond worst-case 
analysis.

Given a set of points (or clients) and a set of candidate centers, the $k$-median problem 
asks for a subset of $k$ candidate centers that minimizes the sum of distances
from each point to its closest center.\footnote{We consider the Euclidean setting
in which the number of candidate centers is polynomial in the number of clients (which is finite
and part of the input). The more general setting where candidate centers can be opened anywhere can be reduced to this with a multiplicative loss in the cost of at most $1+1/\text{poly}(n)$, see~\cite{Mat00}. 
Moreover, even for the two-dimensional 1-median problem, there are known instances where the optimal position 
of the center cannot be described by radicals over the field 
of rationals~\cite{Bajaj1988}, so this assumption is quite common.
} This induces a partitioning of the points
where points in the same group are close to each other.
Such a partitioning finds various applications,
including facility location, image compression~\cite{KMNPSW04},
and community detection. To obtain a more accurate model
of the underlying applications, many variants of the $k$-median problem have 
also been studied. Arguably the most famous are those where the objective functions 
allow the discarding of data points that are irrelevant from the application's perspective.
These variants were introduced by Charikar et al.~\cite{charikar01} 
and referred to as $k$-median with penalties and 
$k$-median with outliers. Another example is the $k$-means objective which 
consists of minimizing the sum of squared distances. This is frequently used in models where 
the goal is to recover mixtures of $k$ Gaussians, a popular problem
in machine learning.

In this paper, we consider $k$ as a fixed parameter and aim at giving tight
upper and lower bounds for the $k$-median, $k$-means, and $k$-median with penalties problems
regarding running time. 

\paragraph{The Parameter $k$:}
The choice of parameterizing by $k$ is a very natural approach when
dealing with issues of tractability. 
Many real-world examples 
involve solving instances of the $k$-median problem 
in low-dimensional Euclidean space. 
A concrete example stemming from machine learning 
is the classic digits dataset (see~\cite{Lichman:2013}), which
consists of images of hand-written digits. Successful approaches for
obtaining a good classification consist of applying an SVD algorithm (\ie singular value decomposition)
to the dataset and solving a three-dimensional $k$-median (or $k$-means) 
instance with $k=10$ (see, \eg~\cite{scikit-learn}).

Other examples include the widely-used hierarchical clustering heuristic
Bisection $k$-means (see~\cite{steinbach2000comparison}), which consists of
recursively dividing a set of points in $d$-dimensional Euclidean space 
using the $k$-median or $k$-means objectives for values of $k \le 10$.

Therefore, as early as the 1990s, the $k$-median and $k$-means problems
have received a great deal of attention from
the theory community, which has tried to 
obtain efficient approximation algorithms for the Euclidean setting. 
Since the work of~\cite{VegaKKR03}, there has been a long line of research 
on $(1+\eps)$-approximation algorithms running in time $f(k,\eps) \text{poly}(n,d)$ for $\epsilon > 0$
(see~\cite{Kumar2004,kumar2005linear,badoiu,FMS07,FeL11,ackermann2008clustering}).
The best algorithm known for $k$-median is due to~\cite{KSS10}, which achieves a $(1+\eps)$-approximation
in time $2^{(k/\eps)^{O(1)}} n \cdot d$, and for $k$-means the best known is due to
Feldman et al.~\cite{FMS07}, which achieves a $(1+\eps)$-approximation in time
$O(nkd + \textrm{poly}(k/\eps)d + 2^{\tilde{O}(k/\eps)})$. While the design of approximation schemes
is fairly well understood for $k$-median and $k$-means when parameterized by $k$, 
the brute-force approach of trying all possible $k$ subsets of candidate centers stubbornly stands as
the best exact algorithm known 
(hence\footnote{Assuming, for the $k$-median upper bound, that fast comparisons of sums of square roots are possible.} a running time of $n^{O(k)}$). 
Thus, obtaining a better bound, even for low-dimensional inputs, is a natural and
important open question. 

This question is further motivated by recent results showing that many famous 
problems (\eg the traveling salesman problem 
or finding a size $k$ independent set of unit spheres)
benefit from the ``$(1-1/d)$ phenomenon,'' namely that there exist exact algorithms 
running in time $n^{O(k^{1-1/d})}$ or $2^{O(n^{1-1/d})}$ (see~\cite{MS14}). 
As Marx and Sidiropoulos showed~\cite{MS14}, this is often tight assuming the Exponential
Time Hypothesis (ETH).
Hence, understanding whether this phenomenon applies to clustering strengthens 
the motivation of studying the $k$-median problem with $k$ as a fixed parameter.

\paragraph{Our Results:}
We show that, quite surprisingly, clustering is hard even in 
Euclidean space of dimension four. Namely,
there is no $f(k)n^{o(k)}$-time algorithm for any computable function $f$ 
for $k$-median or $k$-means unless 
the ETH assumption fails  (Theorem~\ref{thm:pc:4d}). 
For the $k$-median with penalties problem
we show that this hardness bound holds even in $\R^3$ 
(Theorem~\ref{thm:pc:3d}),
and that the hardness becomes $f(k)n^{o(\sqrt{k})}$ in $\R^2$
(Theorem~\ref{thm:pc:2d}).

On the positive side, we give an $n^{O(\sqrt{k})}$-time exact
algorithm in two dimensions for both problems using standard
techniques (Theorem~\ref{thm:upperbound:2d}), and hence provide a
complete characterization of the complexity of the $k$-median with
penalties problem.  Interestingly, this shows a steep gap between the
two-dimensional case and the three-dimensional setting (for $k$-median
with penalties) and the four-dimensional case (for $k$-median).  For
the $k$-median and $k$-median with penalties problems, we assume a
computational model in which sums of square roots can be compared
efficiently, which is a common assumption for geometric problems in
Euclidean space (see for example Gibson et
al.~\cite{gibson2008clustering}).

We note that all of our
results extend to objectives where distances are taken to some power $p$
(for $p=1$ and $p=2$, this yields the $k$-median and $k$-means objectives, respectively).
Moreover, our hardness results do not generalize to versions of the problems where any point
in Euclidean space can serve as a center.  That is, our results only hold for settings where
the set of potential candidate centers is explicitly given as input.

\subsection*{Related Work}
The $k$-median and $k$-means problems
are NP-hard, even in the Euclidean plane (see Meggido and 
Supowit~\cite{MeS84}, Mahajan et al.~\cite{MNV12}, and Dasgupta and 
Freud~\cite{DaF09}). 
This hardness extends to approximation: 
both problems are APX-hard  in the Euclidean setting 
when both $k$ and $d$ are part of the input 
(see Guha and Khuller~\cite{GuK99}, 
Jain et al.~\cite{JMS02}, 
Guruswami et al.~\cite{GI03}, and Awasthi et al.~\cite{ACKS15}).
When $d$ is fixed, however, the problems are no longer APX-hard~\cite{CAKM16,ARR98}.
There has been a large body of work on obtaining 
constant factor approximations for both the $k$-median and 
$k$-means problems (see~\cite{ANSW16,BPRST15,LiS13,JaV01,MeP03}).
The best approximation ratio known for $k$-median in general metric spaces
is due to Byrka et al.~\cite{BPRST15} and is $\approx 2.675$.
For the $k$-means problem, the best known is now 6.357 due to
Ahmadian et al.~\cite{ANSW16}, where they
improved upon the 9-approximation algorithm of Kanungo et al.~\cite{KMNPSW04}.

The literature on fixed-parameter tractability is vast.  We only discuss
the most related works (for a more thorough treatment, see~\cite{CFKLMPPS15}).

\paragraph{Fixed-Parameter Tractability for Fixed $k$.}
There has been a long line of work on $(1+\eps)$-approximation algorithms
for Euclidean $k$-median parameterized by $k$,
\eg~\cite{FeL11,KSS10,HaK07,HaM04}. Many of these works are based
on the notion of a coreset: a representation of the input of size 
poly$(k,\eps)$. There are various algorithms to efficiently compute 
coresets. Once a coreset is computed, the best solution for the coreset
can be found in FPT time (\ie $f(k,\eps)$poly$(n)$).
The best approach known for Euclidean $k$-median runs in time $2^{(k/\eps)^{O(1)}} nd$, due to Kumar 
et al.~\cite{KSS10}.  For Euclidean $k$-means, the best approach known runs in time
$O(nkd + \textrm{poly}(k/\eps)d + 2^{\tilde{O}(k/\eps)})$, due to Feldman et al.~\cite{FMS07}.

\paragraph{Fixed-Parameter Tractability for Fixed $d$.} The choice of $d$ as a parameter
has also been studied. In this case, polynomial time approximation schemes (PTAS) are known 
for both the $k$-median and 
$k$-means problems~\cite{ARR98,CAKM16,FRS16a,KoR07}.
For the $k$-center problem, a lower bound of $n^{o(d)}$ on the running time is known
even when $k = 2$~\cite{cabello2008geometric}. Unfortunately, the 
$k$-center objective (which is a min max objective) 
is quite different from the $k$-median and $k$-means
objectives (which are min sum objectives). Hence, no hardness bound
is known for the $k$-median and $k$-means problems when parameterized
by $d$.

\subsection{Roadmap}
In Section~\ref{sec:prelim}, we introduce some preliminaries. 
In Section~\ref{sec:kmed:genmetric}, we provide some intuition for our main reductions
by giving a simple hardness proof for $k$-median in general metric spaces.
In Section~\ref{sec:pc:3d}, we show hardness of the penalized version of $k$-median in $\mathbb{R}^d$ for $d \geq 3$. 
In Section~\ref{sec:dimhardness}, we show hardness of $k$-median in $\mathbb{R}^d$ for $d\geq 4$.
In Section~\ref{sec:2dhardness}, we show hardness of the penalized version of $k$-median in the two-dimensional case.
Finally, in Section~\ref{sec:2d}, we show an upper bound in the two-dimensional setting for both problems.

\subsection{Overview of Ideas and Techniques}
\paragraph{Lower bounds of the form $f(k)n^{\Omega(k)}$:}
We begin with a straightforward reduction from the Partial
Vertex Cover problem that rules out an $f(k)n^{o(k)}$-time algorithm for $k$-median in general metrics under ETH
(for any computable function $f$). Our observation is the following: obliviously to the parameter $k$, a graph can be represented
as an instance with a candidate for each vertex, and a client for each edge. We set the distance from an edge
to its endpoints to $1$, and its distance to all other vertices to something strictly larger, say, $3$. Then, the number
of covered edges can be read off directly from the cost. 

Unfortunately, the metric example above does not embed well in small dimensions. However, the idea of letting vertices correspond to candidates
and edges to clients can still be made to work.   The first challenge is to place the edges (clients) so that they are closer to their endpoints (candidates)
than to any other candidate. Geometrically, this requires placing the candidates in such a way that the Voronoi cells of any two candidates intersect, so that we can place the clients at the intersections of these cells. Dually, this amounts to finding point sets inducing a Delaunay complex in which its $1$-skeleton is a complete graph. While this is impossible in two dimensions, since the Delaunay complex is a triangulation and is thus sparse, higher dimensions allow for this quite pathological behavior. This is a classic topic in computational geometry (see Erickson~\cite{Erickson2003} and the references therein), and one elegant construction~\cite{seidel1990exact} exhibiting this phenomenon is to place the points on the moment curve $t\mapsto(t,t^2, \ldots, t^d)$, which is what we do in our paper.

For the version with penalties, three dimensions are enough to obtain a lower bound. Here, we prove that for any two values
$t_a,t_b$ that parameterize two vertices $a,b$ (where $(a,b)$ is an edge) on the moment curve,
there is a unique sphere $\mathbb{S}$ tangential to the points on the curve $t = t_a$ and $t=t_b$ that has the entire moment curve exterior
to it. We want to place the client point (corresponding to the edge $(a,b)$) at or near the center of $\mathbb{S}$.
In fact, we can give a little slack, and not consider a tangential sphere, but rather the sphere going through $(t_a,t_a^2,t_a^3)$, $(t_b,t_b^2,t_b^3)$, and two ``dummy''
points on the moment curve placed closely to them.
The difference between the radii of the spheres creates a disparity in the contribution of each covered client (\ie covered edge) to the objective,
which we handle by placing many clients at each center (thus nearly equalizing their contribution). 
Finally, naturally, the associated penalty for an edge is set to be only slightly larger than the radius of the corresponding ball.

For the version without penalties, the task is slightly more challenging:
we need to make sure that each edge is equally costly to ``not cover.'' 
To handle the challenge of uncovered edges, we construct a universal special candidate $z$ that is only slightly farther away from every edge than the two candidates
corresponding to the edge's endpoints. This additional candidate requires us to add an additional dimension to our construction, raising it to four.
Considering the moment curve $m(t)=(t,t^2,t^3,t^4)$ in $\mathbb{R}^4$, the unique sphere through $m(t_z)$ (corresponding to $z$) and tangential
to the later points $m(t_a),m(t_b)$ with $t_a,t_b>t_z$ is such that the moment curve after $m(t_z)$ is exterior to the sphere. 
We may thus choose $z=(1,1,1,1)$ and let all other vertices correspond to points $t>1$. However, placing the edges at the exact centers
of the spheres will not give us any information, as all edges could then be served optimally by $z$. Thus as a final step, we place each
edge near the center, but slightly farther from $z$.

\paragraph{Lower bound in two dimensions:} The lower bounds in two dimensions are a reduction from the Grid Tiling problem using techniques from~\cite{MS14}. 
The main observation is the following: imagine you have uncountably infinitely many clients placed uniformly within a region. If all candidates have the same radius $1$ and the same penalty, then it is always an advantage if the $1$-balls around the chosen candidates overlap as little as possible -- preferably not at all. We can precompute the cost $\nu$ for non-overlapping balls, which is strictly smaller than the cost of any solution where balls overlap. Then, the instance to Grid Tiling has a solution if and only if the constructed $k$-median with penalties
instance has a solution with cost $\leq\nu$. (In fact, exactly $\nu$.)

\paragraph{Upper bound in two dimensions:}
Our upper bounds in two dimensions use the strategy of guessing a separator of size $\sqrt{k}$ in the Voronoi diagram of an optimal solution.
This is quite a useful approach, as illustrated by Marx and Pilipczuk~\cite{MarxP15}.
Since this is quite standard, we defer this result to Section~\ref{sec:2d}.

\section{Preliminaries}\label{sec:prelim}

We frequently use the moment curve throughout our reductions, which we define as follows.

\begin{Definition}
    The curve $\mathbb{R}^+\to \mathbb{R}^d$ defined by $t\mapsto (t,t^2,\ldots ,t^d)$ is called the moment curve.
\end{Definition}

All of our lower bounds are conditioned on the Exponential Time Hypothesis (ETH), which
was conjectured in~\cite{IPZ98}.

\begin{Definition}[Exponential Time Hypothesis(ETH)~\cite{IPZ98}]
There exists a positive real value $s > 0$ such that 3-CNF-SAT,
parameterized by $n$, has no $2^{sn}(n+m)^{O(1)}$-time algorithm (where $n$ denotes the number of variables
and $m$ denotes the number of clauses).
\end{Definition}

The following problem, Partial Vertex Cover, plays a critical role in our reductions.  In particular, we
reduce from this problem to show hardness for $k$-median in $d \geq 4$ dimensions, and $k$-median 
with penalties in $d \geq 3$ dimensions.
\begin{Definition}[Partial Vertex Cover (PVC)]~\\
  \textbf{Input:} A graph $G=(V,E)$, an integer $s \in \mathds{N}$.\\
  \textbf{Parameter:} Integer $k$.\\
  \textbf{Output:}  \texttt{YES} if and only if there exists 
  a set of $k$ vertices  that covers at least $s$ edges.
\end{Definition}

Guo et al.~\cite{PVChard} showed that Partial Vertex Cover is W[1]-hard, but their reduction actually yields a lower bound conditional on ETH. Indeed, they reduced from Independent Set, which is known not to be solvable in time $f(k)n^{o(k)}$ assuming ETH~\cite[Theorem~14.21]{CFKLMPPS15}, and their reduction does not induce blow-up in the size of the parameter. Hence, they actually proved the following lower bound.

\begin{theorem}[PVC Hardness~\cite{PVChard}]
  There is no $f(k)n^{o(k)}$-time algorithm for the Partial Vertex Cover
  problem unless ETH fails (for any computable function $f$), where $n$ is the size of the input.
\end{theorem}
We now give our definitions for the clustering problems we consider in this paper, beginning
with the version without penalties.

\begin{Definition}[$d$-Dimensional $k$-Median]~\\
\textbf{Input:} A set of candidate centers $C \subset \R^d$,
  a set of clients $A \subset \R^d$, a cost $\nu \in \mathbb{Q}$.\\
  \textbf{Parameter:} Integer $k$.\\
  \textbf{Output:} \texttt{YES} if and only if there exists a set
  $S$ of $k$ candidate centers such that
  $$\sum_{a \in A} d(a, S) \le \nu.$$
  \label{def:kmed}
\end{Definition}
Here, the distance of a point $a \in \R^d$ to a set $S$ is the minimum distance from $a$
to any point in the set $S$ (\ie $d(a,S) = \min_{c \in S} d(a,c)$).
Unless stated otherwise, we use $n$ to denote the size of the input to the problem.
In addition, we note that our results extend to objective functions where distances are taken to some power $p$,
namely $d(a,S)^p$.  The important special cases of $p=1$ and $p=2$ yield the $k$-median and $k$-means objectives, respectively.
We now consider a slightly more general version of the $k$-median problem,
see also~\cite{charikar01} for previous definitions.

\begin{Definition}[$d$-Dimensional $k$-Median with Penalties]~\\
  \textbf{Input:} A set of candidate centers $C \subset \R^d$,
  a set of clients $A \subset \R^d$, a penalty $p_a$ for each $a \in A$,
  a cost $\nu \in \mathbb{Q}$.\\
  \textbf{Parameter:} Integer $k$.\\
  \textbf{Output:} \texttt{YES} if and only if there exists a set
  $S$ of $k$ candidate centers such that
  $$\sum_{a \in A} \min(d(a, S),p_a) \le \nu.$$
\end{Definition}
In our reductions, we sometimes set the cost threshold $\nu$ to be an irrational number.
We can remedy this issue since, in our reductions, there is always a large gap between the most costly
yes-instances and the least costly no-instances (in particular, the gap is at least an inverse polynomial).
Hence, we can always choose a rational number strictly larger than $\nu$ (but smaller than the least
costly no-instance) such that our reductions take time polynomial in the size of the input.
By size of the input, we refer to the number of bits it takes to represent the candidate centers,
client points, and cost bound (for the penalty version, the size of the input also includes the
bits used to represent the penalty amounts).

\subsection{Properties of the Moment Curve}
We first prove the following
property regarding ($3$-)spheres and the moment curve, which will be useful in our reduction. The proofs, in particular the use of Descartes' rule of signs~\cite{curtiss}, follow the exposition of Edelsbrunner~\cite[Section~4.5]{edelsbrunner2014short}.

\begin{lemma}\label{L:descartes}
Fix any $5$ positive values $0 < t_1 < t_2 < t_3 < t_4 < t_5$, and consider the corresponding $5$ points
that lie on the moment curve given by $(t_i,t_i^2,t_i^3,t_i^4)$ for $1 \leq i \leq 5$.  Then the unique $3$-sphere
that goes through these $5$ points satisfies the following property: the segments on the moment curve
corresponding to $t \in (t_1,t_2) \cup (t_3,t_4) \cup (t_5, \infty)$ all lie outside of the sphere (\ie
the distance of all such points from the center of the $3$-sphere is strictly more than its radius).
\end{lemma}
\begin{proof}
Consider any such set of $5$ positive values $t_i >0$ and their corresponding $5$ points on the moment curve given by $p_i = (t_i,t_i^2,t_i^3,t_i^4)$ for $1 \leq i \leq 5$.  These $5$ points
on the moment curve define a unique $3$-sphere in $\R^4$, with center $(a,b,c,d)$ and radius $r$.
Consider the following function given by $f(t) = (t - a)^2 + (t^2 - b)^2 + (t^3 - c)^2 + (t^4 - d)^2 - r^2$.
Observe that the roots of this polynomial correspond to values of the parameter $t$ where the moment curve
intersects the $3$-sphere.  Moreover, since the points $p_i$ lie on the moment curve and on the $3$-sphere by construction, we have
$f(t_i) = 0$ for all $1 \leq i \leq 5$ (\ie each $t_i$ is a root of $f(t)$).

We consider applying Descartes' rule of signs, which we will use to upper bound the number of strictly positive roots of $f(t)$.
The rule says that the number of strictly positive roots of a polynomial is upper bounded by the number of sign changes between non-zero
coefficients (assuming the coefficients are arranged in decreasing order of the degree of their corresponding term).  To this end,
we expand the polynomial $f(t)$:
\begin{align*}
  f(t) &= t^2 - 2at + a^2 + t^4 - 2bt^2 + b^2 + 
  t^6 - 2ct^3 + c^2 + t^8 - 2dt^4 + d^2 - r^2\\
     &= t^8 + t^6 + (1-2d)t^4 - 2ct^3 + (1-2b)t^2 - 2at + (a^2 + b^2 + c^2 + d^2 - r^2).
\end{align*}
Hence, the coefficient sequence is given by $(1,1,(1-2d),-2c,(1-2b),-2a,(a^2+b^2+c^2+d^2-r^2))$.  Clearly, there are (at most)
$5$ changes in sign in this sequence, which implies the number of strictly positive roots is upper bounded by $5$.
However, we already know of $5$ roots to this polynomial, and hence the only places where the moment curve intersects
the $3$-sphere for positive values of $t$ are for $t=t_i$.

In particular, since Descartes' rule of signs counts roots of multiplicity separately, the moment curve is not
tangent to the sphere for any $t > 0$.  Now, consider the moment curve in the open interval $(t_5,\infty)$.  It must
be the case that the entire curve in this interval lies outside the $3$-sphere.  If not, it would have to exit the
sphere again at some point, which would result in an additional root (a contradiction).  In the following, we imagine
going along the curve backwards (\ie for decreasing values of the parameter $t$).  For the open interval $(t_4,t_5)$,
since the moment curve is not tangent to the sphere at $t = t_5$, it must go inside the sphere.  The next time the curve
intersects the $3$-sphere is at $t = t_4$, and hence the curve lies inside the $3$-sphere in the open interval $(t_4,t_5)$.
Similarly, since the curve is not tangent at $t=t_4$, it must exit the $3$-sphere at $t=t_4$ and then
intersect the $3$-sphere next at $t=t_3$, implying that the curve lies outside of the $3$-sphere in the open interval $(t_3,t_4)$.
Using the same reasoning, we conclude that the $3$-sphere lies completely inside the $3$-sphere in the open interval $(t_2,t_3)$,
and then completely outside of the $3$-sphere in the open interval $(t_1,t_2)$, giving the lemma.

\end{proof}

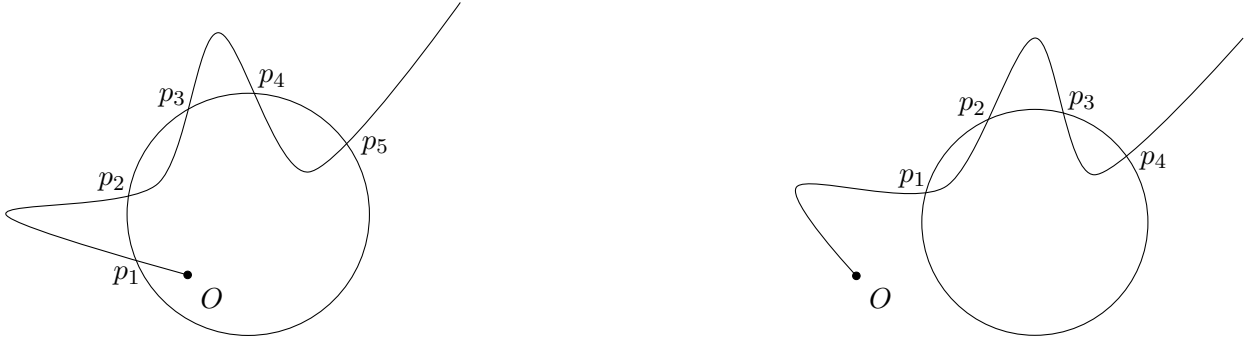
\begin{figure*}[ht!]
\begin{tikzpicture*}{0.37\textwidth}
\begin{scope}[
  vertex/.style={
  draw,
  circle, fill,
  minimum size=1mm,
  inner sep=0pt,
  outer sep=0pt%
},
]

\node[vertex,label={below right:$O$}] (o) at (0,0) {};
\node[label={$p_1$}] (o) at (-1,-0.5) {};
\node[label={$p_2$}] (o) at (-1.25,1.05) {};
\node[label={$p_3$}] (o) at (-0.25,2.45) {};
\node[label={$p_4$}] (o) at (1.4,2.75) {};
\node[label={$p_5$}] (o) at (3.1,1.65) {};
\draw plot [smooth] coordinates {(0,0) (-3,1) (-0.5,1.5) (0.5,4) (2,1.7) (4.5,4.5)}; 
\draw (1,1) circle (2cm);

\end{scope}
\end{tikzpicture*}
\hspace{0.25\textwidth}
\begin{tikzpicture*}{0.37\textwidth}
    \begin{scope}[
        vertex/.style={
            draw,
            circle, fill,
            minimum size=1mm,
            inner sep=0pt,
            outer sep=0pt%
        },
        ]
        
        \node[vertex,label={below right:$O$}] (o) at (0,0) {};
        \node[label={$p_1$}] (o) at (0.95,1.15) {};
		\node[label={$p_2$}] (o) at (1.95,2.35) {};
		\node[label={$p_3$}] (o) at (3.8,2.45) {};
		\node[label={$p_4$}] (o) at (5,1.45) {};
        \draw plot [smooth] coordinates {(0,0) (-1,1.5) (1.5,1.5) (3,4) (4,1.7) (6.5,4)}; 
        \draw (3,0.9) circle (1.9cm);
    \end{scope}
\end{tikzpicture*}
%\label{fig:curve}
    \caption{In $\mathbb{R}^4$ (left), the unique $3$-sphere through the points $p_1,\ldots,p_5$ on the moment curve has no other intersections with the moment curve
after the origin. 
        In $\mathbb{R}^3$ (right), the unique sphere through the points $p_1,\ldots, p_4$ on the moment curve has no other intersections with the moment curve.
\label{fig:sphereandcurve}}
\end{figure*}

We now prove (in a very similar manner) an analogous result for spheres in $\R^3$.  In the following, we denote by $O$ the origin.
\begin{lemma}\label{lem:moment:3d}
    Fix any $4$ positive values $0 < t_1 < t_2 < t_3 < t_4$, and consider the corresponding $4$ points
    that lie on the moment curve given by $(t_i,t_i^2,t_i^3)$ for $1 \leq i \leq 4$.  
    Then the unique sphere that goes through these $4$ points satisfies the following property: the segments on the moment curve
    corresponding to $t \in (O,t_1) \cup (t_2,t_3) \cup (t_4, \infty)$ all lie outside of the sphere (\ie
    the distance of all such points from the center of the sphere is strictly more than its radius).
\end{lemma}
\begin{proof}
Similarly to the proof of Lemma~\ref{L:descartes}, let $\mathbb{S}_r(a,b,c)$ be the unique sphere with center $(a,b,c)$ and radius $r$ through
the points. We then analyze the function
\begin{align*}
 f(t) &= (t-a)^2 + (t^2-b)^2 + (t^3 - c)^2 - r^2\\
 &= t^2 - 2at + a^2 + t^4 - 2bt^2 + b^2 + t^6 - 2ct^3 + c^2 - r^2\\
 &= t^6 + t^4 - 2ct^3 + (1-2b)t^2 -2at +(a^2+b^2+c^2-r^2).
\end{align*}
The coefficients are $(1,1,-2c,(1-2b),-2a,(a^2+b^2+c^2-r^2))$, which has (at most) $4$ changes of sign, which by Descartes' rule means that there are at most $4$ roots. But then, since $p_1,\ldots,p_4$ already constitute $4$ roots, there are no other roots. 
Then, the segment $(t_4, \infty)$ of the moment curve must lie entirely outside the sphere. 
Furthermore, since the roots are counted with multiplicity, the section $(t_3,t_4)$ lies inside the sphere, the section $(t_2, t_3)$ lies outside the sphere, the section $(t_1, t_2)$ lies inside the sphere, and, finally, the section $(O,t_1)$ lies outside the sphere.
\end{proof}

\section{Warm-up: Hardness of $k$-Median for General Metric Spaces}\label{sec:kmed:genmetric}
In this section, we show that assuming ETH, there is no $f(k)n^{o(k)}$-time exact
algorithm for $k$-median in general metric spaces (for any computable function $f$).
In Section~\ref{sec:dimhardness}, we show how to make this reduction work in $\R^4$.

\begin{theorem}
  \label{thm:kmed:genmetric}
  There is no $f(k)n^{o(k)}$-time algorithm that solves the $k$-median problem in 
  general metric spaces unless ETH fails (for any computable function $f$), where $n$
  is the size of the input.
\end{theorem}

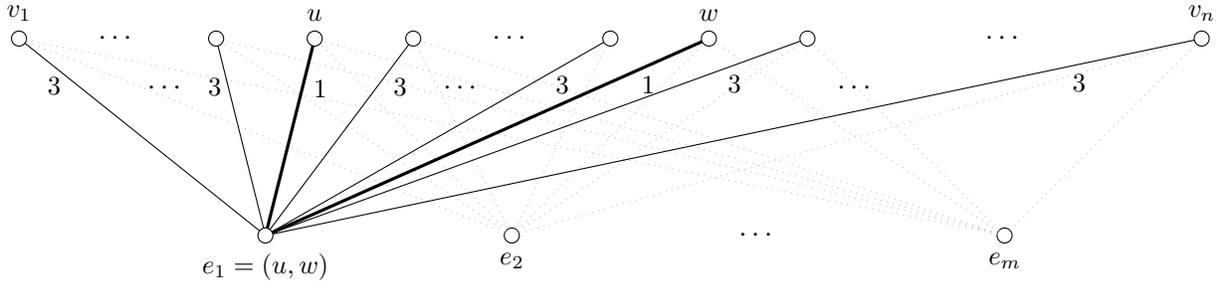
\begin{figure*}[ht!]
    \begin{tikzpicture*}{\textwidth}
        \begin{scope}[
            vertex/.style={
                draw,
                circle,
                minimum size=2mm,
                inner sep=0pt,
                outer sep=0pt%
            },
            vedge/.style={
                very near end,
                below,
                rotate=0,
                font=\small,
            },
            dedge/.style={
                near start,
                below,
            },
            every label/.append style={
                rectangle,
                %label distance=.1cm,
                font=\small,
            }
            ]
            \node[vertex,label={below:$e_1 = (u,w)$}] (c1) at (2.5,0) {};
            \node[vertex,label={below:$e_2$}] (c2) at (5,0) {};
            \node at (7.5,0) {$\cdots$};
            \node[vertex,label={below:$e_m$}] (c3) at (10,0) {};

            \node[vertex,label={above:$v_1$}] (s0) at (0,2) {};
            \node at (1,2) {$\cdots$};
            \node[vertex,label={above:}] (sj) at (2,2) {};
            \node[vertex,label={above:$u$}] (u) at (3,2) {};
            \node[vertex,label={above:}] (s1) at (4,2) {};
            \node at (5,2) {$\cdots$};
            \node[vertex,label={above:}] (s2) at (6,2) {};
            \node[vertex,label={above:$w$}] (w) at (7,2) {};
            \node[vertex,label={above:}] (s3) at (8,2) {};
            \node at (10,2) {$\cdots$};
            \node[vertex,label={above:$v_n$}] (s4) at (12,2) {};

            \draw[color=lightgray,dotted] (c2) -- (s0) node[dedge] {};
            \draw[color=lightgray,dotted] (c2) -- (s1) node[dedge] {};
            \draw[color=lightgray,dotted] (c2) -- (sj) node[dedge] {};
            \draw[color=lightgray,dotted] (c2) -- (u) node[dedge] {};
            \draw[color=lightgray,dotted] (c2) -- (w) node[dedge] {};
            \draw[color=lightgray,dotted] (c2) -- (s2) node[dedge] {};
            \draw[color=lightgray,dotted] (c2) -- (s3) node[dedge] {};
            \draw[color=lightgray,dotted] (c2) -- (s4) node[dedge] {};
            
            \draw[color=lightgray,dotted] (c3) -- (s0) node[dedge] {};
            \draw[color=lightgray,dotted] (c3) -- (s1) node[dedge] {};         
            \draw[color=lightgray,dotted] (c3) -- (sj) node[dedge] {};
            \draw[color=lightgray,dotted] (c3) -- (u) node[dedge] {};
            \draw[color=lightgray,dotted] (c3) -- (w) node[dedge] {};
            \draw[color=lightgray,dotted] (c3) -- (s3) node[dedge] {};
            \draw[color=lightgray,dotted] (c3) -- (s4) node[dedge] {};
            \draw (c1) -- (s0) node[vedge] {$3$};
            \draw (c1) -- (sj) node[vedge] {$3$~~~};
            \draw[very thick] (c1) -- (u) node[vedge] {~~~$1$};
            \draw[very thick] (c1) -- (w) node[vedge] {$1$};
            \draw (c1) -- (s1) node[vedge] {~~$3$};
            \draw (c1) -- (s2) node[vedge] {$3$};
            \draw (c1) -- (s3) node[vedge] {$3$};
            \draw (c1) -- (s4) node[vedge] {$3$};
            
            \node at (1.5,1.5) {$\cdots$};
            \node at (4.5,1.5) {$\cdots$};
            \node at (8.5,1.5) {$\cdots$};
            
        \end{scope}
    \end{tikzpicture*}
    \vspace{-0.25cm}
    \caption{The distance from $(u,w)$ to $u$ and to $w$ is $1$, and to all other vertices it is $3$.}
    \label{fig:metric}
\end{figure*}

We now describe the reduction (see Figure~\ref{fig:metric}). Let $G=(V,E)$, $s$, and $k$ be an instance of
PVC.  We denote by $m$ the number of edges, namely $m = |E|$.  We build the following metric space: for each vertex $v \in V$ we
create a point $x_v$. For each edge $(u,v) \in E$ we create a point
$y_{(u,v)}$.
The distances are the following: for each $x_z,y_{(u,v)}$, we have $d(x_z, y_{(u,v)}) = 1$ if $z \in \{u,v\}$ or $3$ if $z \notin \{u,v\}$. Finally, the remaining distances are given by the shortest path metric
induced by the distances already defined.

We now define an instance of the $k'$-median problem.
We let $k' = k$, $C = \{ x_u \mid u \in V\}$,
$A = \{ y_{(u,v)} \mid (u,v) \in E\}$, and $\nu = s + 3(m-s)$.

We show the following claim, which implies Theorem~\ref{thm:kmed:genmetric}.
\begin{claim}
  $G$ has a PVC with $k$ vertices covering at least $s$ edges if and only
  if there exists a solution to the $k'$-median instance of cost
  at most $s+3(m-s)$.
\end{claim}
\begin{proof}
  Consider first an instance of PVC with $k$ vertices $\{v_1,\ldots,v_k\}$ covering at least $s$ edges.
  We claim that the solution to the $k'$-median instance in which we open the $k'$ candidates given by
  $S_0 = \{x_{v_1},\ldots,x_{v_k}\}$ has cost at most $\nu$. Observe that for the edges $(u,v)$
  covered in the PVC solution, the points $y_{(u,v)}$ are at distance exactly one from a center
  of $S_0$. Moreover, the points $y_{(u,v)}$ that correspond to an edge $(u,v)$ that are not covered
  by the PVC solution are at distance exactly 3 from a center of~$S_0$. Since there are at most
  $m-s$ such points, we have a $k'$-median solution of cost at most $s + 3(m-s)$.

  Now consider a $k'$-median solution given by $S_0 = \{x_{v_1},\ldots,x_{v_k}\}$ of cost at most $s+3(m-s)$.
  By definition, each point is at distance either 1 or 3 from a center and the total number of points is~$m$.
  It follows that the total number of points at distance 1 is at least $s$. Each such point represents an edge
  that has an endpoint in the set $\{v_1,\ldots,v_k\}$.
  Thus, the vertices $v_1,\ldots,v_k$ induce a partial vertex cover of size $k$ covering at least $s$ edges.
\end{proof}

\section{Hardness of $k$-Median with Penalties in Three Dimensions}
\label{sec:pc:3d}

In the following two sections, we establish $f(k)n^{o(k)}$-time lower bounds
for the $k$-median problem in Euclidean spaces of low dimension (for any computable function~$f$).
It seems easier to establish hardness for the $k$-median problem with
penalties, and thus our first result is Theorem~\ref{thm:pc:3d}, which
works in any dimension of at least three.  We first give details on the
reduction before proving structural properties.  We follow the
same structure as in the proof of Theorem~\ref{thm:kmed:genmetric}
in Section~\ref{sec:kmed:genmetric}. Namely, we reduce from
Partial Vertex Cover and create a candidate center for each vertex of
the input graph, along with a client for each edge of the input graph
$G=(V,E)$. For each edge, we ensure that the client corresponding to
that edge is closer to the two centers representing the endpoints of
the edge than to any other candidate center. This is the key property
in our reduction and we show how it can be satisfied in $\R^3$ for
all edges of the input graph.

\begin{theorem}
  \label{thm:pc:3d}
  There is no $f(k)n^{o(k)}$-time exact algorithm for the $3$-dimensional
  $k$-median with penalties problem, unless ETH fails (for any computable function $f$),
  where $n$ is the size of the input.
\end{theorem}

We now provide the location of the candidate centers created.
For each vertex $v_i$ of the input graph, we create a candidate 
center $\tilde{v}_i$. We call $v_i$ the corresponding vertex of $\tilde{v}_i$.
We place the candidate centers on the moment 
curve: the candidate center $\tilde{v}_i$ is placed
at $(2i,(2i)^2,(2i)^3)$.
We also associate a dummy point $d_i$ with each candidate center $\tilde{v}_i$.
The $d_i$ are not part of the $k$-median instance 
and only used to generate the instance.
We place $d_i$ at $(2i+1, (2i+1)^2, (2i+1)^3)$. Let $C$ be the set of candidate centers  
$\{(2i,(2i)^2,(2i)^3) \mid i \in \{1,\ldots,|V|\}\}$
and let $C^+ = C \cup \bigcup_i \{d_i\}$. By construction, we have the following fact: for all $i$,
there is no candidate on the moment curve between $\tilde{v}_i$ (\ie $t = 2i$) and $d_i$ (\ie $t = 2i+1$).

We now explain how to create client points that correspond to edges.
Let $e_{i,j} = (v_i,v_j)$ be an edge of~$G$, of which there are $m = |E|$. Since points on the moment curve are in general position, there is a unique sphere $\mathbb{S}_{i,j}$
that intersects the moment curve at the points $\tilde{v}_i,d_i,\tilde{v}_j,d_j$, the center and radius of which we denote by $c_{i,j}$
and $r_{i,j}$, respectively.  By Lemma~\ref{lem:moment:3d}, we know that there is no point $p \in C^+ - \{\tilde{v}_i,d_i,\tilde{v}_j,d_j\}$
that is contained in the ball of center $c_{i,j}$ and radius $r_{i,j}$.

Let $q$ be an index pair that gives rise to the maximum radius $r_{i,j}$, namely $q=\argmax_{i,j} (r_{i,j})$ (\ie $q$ is of the form ``$i,j$").
We also let $\delta>0$ be some inverse polynomially small fraction to be defined (\ie $\delta = \frac{1}{|V|^c}$ for some constant $c > 0$).
We place $n_q = \lceil \frac{1}{\delta} \rceil$ client points
at $c_q$, and $n_{i,j} = \lceil n_q \frac{r_q}{r_{{i,j}}}\rceil$ client points at all other centers $c_{i,j} \neq c_q$.
We let $\cost_{i,j} = n_{i,j} \cdot r_{i,j}$ and $\mu = n_q \cdot \maxrad$.

\begin{lemma}\label{lem:3d-discretization}
For any pair $i,j$ such that $(v_i,v_j)$ is an edge, $cost_{i,j}$ satisfies $\mu \leq cost_{i,j} \leq (1+\delta)\mu$.
In addition, $\mu$ and $n_{i,j}$ (corresponding to each center $c_{i,j}$) are polynomially bounded in $|V|$.
\end{lemma}
\begin{proof}
Fix any such pair $i,j$.  Clearly, the claim is true for $c_{i,j} = c_q$ (since $cost_q = n_q \cdot r_q$), so consider any such center $c_{i,j} \neq c_q$.
For the first inequality (\ie the lower bound), we have the following:

\[ cost_{i,j} = r_{{i,j}}\cdot n_{{i,j}} = r_{{i,j}} \left\lceil \frac{r_q}{r_{{i,j}}} \cdot n_q\right\rceil \geq r_{{i,j}} \cdot \frac{r_q}{r_{{i,j}}} \cdot n_q = r_q \cdot n_q = \mu. \]

For the second inequality (\ie the upper bound), we get:

\begin{align*}
cost_{i,j} = r_{{i,j}}\cdot n_{{i,j}} &= r_{{i,j}} \left\lceil \frac{r_q}{r_{{i,j}}} \cdot n_q\right\rceil \leq r_{{i,j}}\left(\frac{r_q}{r_{{i,j}}} \cdot n_q + 1\right) = r_q \cdot n_q + r_{{i,j}} \leq r_q \cdot n_q + r_q\\
&\leq r_q \cdot n_q + r_q \cdot \delta \left\lceil \frac{1}{\delta}\right\rceil = r_q \cdot n_q + r_q \cdot \delta \cdot n_q = (1+\delta)r_q \cdot n_q = (1+\delta)\mu.
\end{align*}
To obtain our polynomial bound claims, we first note that $n_q$ is polynomially bounded since $\delta$ is
an inverse polynomial.  To argue that $n_{i,j}$ is polynomially bounded for all other centers $c_{i,j} \neq c_q$,
it suffices to upper bound $r_q$ and lower bound $r_{i,j}$.
We observe that for any edge $(v_i,v_j)$, $c_{i,j}$ is the circumcenter of four points, and is thus the
intersection of three hyperplanes (the perpendicular bisectors of these points). Therefore, it is the
solution of a linear system of equations of constant dimension with entries that are integers or half-integers,
because the points $\tilde{v}_k$ and $d_k$ have integer coordinates. It follows that $c_{i,j}$ has coordinates
described by a constant degree rational fraction of the coordinates of the points $\tilde{v}_k$ and $d_k$.
Therefore the maximal radius is polynomially bounded, and similarly, the radii $r_{i,j}$ cannot be exponentially small.
Finally, note that $\mu$ must also be polynomially bounded, since it is the product of two polynomials, namely $n_q$ and $\maxrad$.
\end{proof}

We let the set of client points $A$ be the set of all copies of all $c_{i,j}$.
We now define the price of a copy of $c_{i,j}$ to be
$p_{i,j} = r_{i,j} + \eps/n_{i,j}$ for some small enough constant $\eps$.
Let $P$ denote the set of prices (\ie penalties).
For a small enough $\eps$, the following fact follows from Lemma~\ref{lem:moment:3d}.
\begin{fact}
  \label{fact:price}
  For any $c_{i,j}$, and for any solution $S$ such that $\tilde{v}_i,\tilde{v}_j \not\in S$,
  we have $\dist(c_{i,j}, S) > r_{i,j} + \eps/n_{i,j}$.
\end{fact}

It follows that the cost of serving the copies of $c_{i,j}$ in
a solution $S$ such that $\tilde{v}_i,\tilde{v}_j \notin S$ is $\price_{i,j} = \cost_{i,j} + \eps$.
Moreover, we have that for any solution, either all the copies of
$c_{i,j}$ are served by $\tilde{v}_i$ or $\tilde{v}_j$ or they are all paying a price 
$\price_{i,j}$.
\begin{lemma}
  \label{lem:3d:struct}
  For any $k$-median solution $S$, and
  for any $c_{i,j}$, we have that the cost induced by 
  the copies of $c_{i,j}$ is:
  \begin{itemize}
  \item $\cost_{i,j}$ if $\tilde{v}_i \in S$ or $\tilde{v}_j \in S$,
  \item $\cost_{i,j} + \eps$ otherwise.
  \end{itemize}
\end{lemma}
\begin{proof}
  If $\tilde{v}_i \in S$ or $\tilde{v}_j \in S$, we have by 
  Lemma~\ref{lem:moment:3d}, that the distance from any copy
  of $c_{i,j}$ to $S$ is $r_{i,j} < p_{i,j}$. Therefore
  the cost induced by each copy of $c_{i,j}$ is $r_{i,j}$ and
  so the total cost is $\cost_{i,j}$.

  Now, if $\tilde{v}_i,\tilde{v}_j \notin S$, we have by Fact~\ref{fact:price}
  that the cost induced by each copy of $c_{i,j}$ is given by
  $\min(\dist(c_{i,j},S), p_{i,j}) = p_{i,j}$.
  Therefore, we conclude that the total cost induced by the 
  copies of $c_{i,j}$ is $n_{i,j}p_{ij} = \cost_{i,j} + \eps$.
\end{proof}

We can now complete the proof of the theorem.
\begin{proof}[Proof of Theorem~\ref{thm:pc:3d}]
First, by Lemma~\ref{lem:3d-discretization}, we have that the size of the 
instance is $|V|^{O(1)}$.

We show that the answer to the $k$-median instance $(C,A,P,\nu)$ 
described above, where 
$\nu = (1+\delta)\left( \mu \cdot s + (m-s)(\mu + \eps)\right)$, is 
\texttt{YES}
if and only if there exists a partial vertex cover with $k$ vertices covering at least $s$ edges.

First, if there exists such a partial vertex cover, we claim that we can pick the 
$k$ candidate centers corresponding to the $k$ vertices and obtain 
a solution
of cost at most $\nu$. Indeed, by Lemma~\ref{lem:3d-discretization},
each set of clients corresponding to an edge 
$(v_i,v_j)$ can be served by $\tilde{v}_i$ or $\tilde{v}_j$ and induces a cost of at most 
$\cost_{i,j} \le (1+\delta) \mu$.
Each set of clients corresponding to an edge $(v_i,v_j)$ not covered 
induces a cost of 
$\price_{i,j} \le (1+\delta) \mu +\eps$. It follows that the 
induced solution to
the $k$-median problem has cost at most $\nu$.

Now assume that there is a solution $S$ of cost at most 
$\nu$
to the $k$-median problem on the instance $(C,A,P, \nu)$.
By Lemma~\ref{lem:3d:struct}, for each $c_{i,j}$, 
we have that either all the copies of $c_{i,j}$
are served by a single center which is either $\tilde{v}_i$ or $\tilde{v}_j$ or all of 
them are paying a price $p_{i,j}$. 
It follows that for each $c_{i,j}$ such that $\tilde{v}_i,\tilde{v}_j \notin S$, 
the cost induced by the 
copies of $c_{i,j}$ is at least 
$\mu + \eps$.

We now argue that at least $s$ pairs $i,j$ (corresponding to $c_{i,j}$)
are being served either by $\tilde{v}_i$ or $\tilde{v}_j$.
We denote by $E_1$ the set of edges $e_{i,j}$ for which a candidate center is open at one of $\tilde{v}_i$ or $\tilde{v}_j$.
Then the cost of the $k$-median instance is 

\[\sum_{e_{i,j} \in E_1} cost_{i,j} + \sum_{e_{i,j} \in E\setminus E_1} (cost_{i,j} + \epsilon) \geq \mu|E_1|+ (m-|E_1|)(\mu+\varepsilon),\]
where the inequality comes from Lemma~\ref{lem:3d-discretization}. By hypothesis, the cost is bounded by $\nu$, which means:

\[
\mu |E_1| + (m-|E_1|)(\mu + \epsilon) \leq (1+\delta)(\mu m + \epsilon(m-s)) \iff
|E_1| \geq s - \frac{\delta m \mu}{\epsilon} - \delta(m-s).
\]
As long as the last expression, $s - \frac{\delta m \mu}{\epsilon} - \delta(m-s)$, is strictly more than $s-1$, then we can conclude that
$|E_1| > s - 1$. Since $|E_1|$ is an integer, this would yield our desired bound $|E_1| \geq s$.  For $\delta < \frac{\epsilon}{m(\mu + \epsilon)}$,
this holds.  Note that since $m$ and $\mu$ are polynomially bounded (by Lemma~\ref{lem:3d-discretization}), $\delta$ can be taken
to be an inverse polynomial.
Hence, taking the vertices corresponding to the centers of the $k$-median
solution yields a partial vertex cover consisting of $k$ vertices covering at least $s$ edges.
\end{proof}

\section{Hardness of $k$-Median in Four Dimensions}\label{sec:dimhardness}

In this section, we prove that for any fixed $d \geq 4$, and for any
fixed $k$, there does not exist an $f(k)n^{o(k)}$-time algorithm that solves
$k$-median in $d$-dimensional space exactly, unless ETH fails (for any computable function~$f$).  Our
proof is similar in spirit to the reduction given as a warm-up in
Section~\ref{sec:kmed:genmetric}, and even more similar to the one of
Theorem~\ref{thm:pc:3d}, yet the absence of penalties makes the reduction more delicate.  We only prove our hardness result for $d=4$
dimensions, which in turn implies our result for dimensions larger than $4$.

\begin{theorem}
  \label{thm:pc:4d}
  For any dimension $d \geq 4$, there is no $f(k)n^{o(k)}$-time exact algorithm for the $d$-dimensional
  $k$-median problem, unless ETH fails (for any computable function $f$), where $n$ is the size of the input.
\end{theorem}

For a fixed parameter $k$, we are given a graph $G = (V,E)$ on $n = |V|$ vertices and $m = |E|$ edges,
along with an integer $s$.  Arbitrarily index the vertices $v_1,\ldots,v_n$.
We construct a $k'$-median instance with candidate set $C$, client set $A$, and cost bound $\nu$
as follows.  We let $k' = k+1$, and consider the moment curve $(t,t^2,t^3,t^4)$.  In particular, we add $n+1$ candidate points to $C$, which
all lie on the moment curve.  
There is one special candidate center, which we denote by $z^*$, placed on the curve at $t=1$ (\ie
$z^* = (1,1,1,1)$).  For each vertex $v_i$, we add a candidate center on the curve at $t = 2i$ for $1 \leq i \leq n$
(\ie $(2i, (2i)^2, (2i)^3, (2i)^4)$), denoted by $\tilde{v}_i$.

For each edge $e_{i,j} = (v_i,v_j)$ in $G$, consider the unique $3$-sphere, which we denote by $\mathbb{S}_{i,j}$, defined by the following $5$ points:
$z^*$, $\tilde{v}_i$, $\tilde{v}_j$, and the two points on the moment curve given by $t = 2i+1$ and $t = 2j+1$.
Let $c_{i,j}$ and $r_{i,j}$ denote the center and radius of the $3$-sphere $\mathbb{S}_{i,j}$,
respectively.  In the following, we slightly perturb the center $c_{i,j}$ of each such sphere such that it remains equidistant to $\tilde{v}_i$
and $\tilde{v}_j$ (though farther away from $z^*$), and denote the new (perturbed) position by $c'_{i,j}$, and the corresponding distance
to $\tilde{v}_i$ and $\tilde{v}_j$ by $r'_{i,j}$.

\begin{lemma}\label{L:perturbation}
  There exists $\varepsilon>0$ such that for all $i,j$ where $e_{i,j}=(v_i,v_j) \in E$, there is a point $c'_{i,j}$ such that:
  \begin{itemize}
\item $r'_{i,j} := d(c'_{i,j},\tilde{v}_i)=d(c'_{i,j},\tilde{v}_j)$, 
\item $d(c'_{i,j},z^*)= (1+\varepsilon)r'_{i,j}$, and
\item for all $k\neq i,j$, $d(c'_{i,j},\tilde{v}_k) \geq (1+\varepsilon)r'_{i,j}$.
\end{itemize}
\end{lemma}

\begin{proof}
First, let us observe that by Lemma~\ref{L:descartes}, for an edge $e_{i,j}=(v_i,v_j)$, the ball centered at $c_{i,j}$ of radius $r_{i,j}$ contains no candidate center in its interior, and only $z^*$, $\tilde{v}_i$, and $\tilde{v}_j$ on its boundary.

  The strategy of the proof is to perturb $c_{i,j}$ in a very small ball to obtain $c'_{i,j}$.  Since the number of points $\tilde{v}_k$ is bounded, and for any $k\neq i,j$, $d(c_{i,j},\tilde{v}_k)>r_{i,j}$, there exists $\eta>0$ such that for all $k\neq i,j$, $d(c_{i,j},\tilde{v}_k)>(1+\eta)r_{i,j}$. Therefore, any point in a ball centered at $c_{i,j}$ of radius $r\leq r_{i,j} \eta/2$ is at distance at least $(1+\eta/2)r_{i,j}$ from any $\tilde{v}_k$ for $k \neq i,j$. Now, we consider the intersection of such a small ball with the $3$-dimensional hyperplane $H$ equidistant to $\tilde{v}_i$ and $\tilde{v}_j$. In this $3$-dimensional space, the inequality $d(x,\tilde{v}_i)<d(x,z^*)$ defines a $3$-dimensional subspace that is nonempty (because $z$ is different from $\tilde{v}_i$ and $\tilde{v}_j$) from which we take a point $c'_{i,j}$ such that $d(c'_{i,j}, z^*)\leq (1+\eta/2)r_{i,j}$. This can be done since we can take it arbitrarily close to $c_{i,j}$. Finally, this can be done consistently for all the edges $(v_i,v_j)$, so that for all of these, there is an $\varepsilon>0$ such that $d(c'_{i,j},z^*)= (1+\varepsilon)r'_{i,j}$. This proves the lemma.

\end{proof}

Let $q$ be an index pair that gives rise to the maximum $r'_{i,j}$, namely $q = \argmax_{{i,j}}r'_{{i,j}}$ (\ie $q$ is of the form ``$i,j$").
Let $\delta>0$ be some inverse polynomially small fraction to be defined, $n_q = \lceil \frac{1}{\delta} \rceil$, and $n_{i,j} = \lceil n_q \frac{r'_q}{r'_{{i,j}}}\rceil$ for
all ${i,j} \neq q$. We place $n_{i,j}$ client points at $c'_{{i,j}}$ for each edge $(v_i,v_j)$.
Finally, we place $n_{z^*}=|E|n_qr'_q$ client points at $z^*$. We write $cost_{i,j}=n_{i,j}r'_{i,j}$ and $\mu=n_qr'_q$. 

\begin{lemma}\label{L:discretization}
For any pair $i,j$ such that $(v_i,v_j)$ is an edge, $cost_{i,j}$ satisfies $\mu \leq cost_{i,j} \leq (1+\delta)\mu$.
\end{lemma}

 \begin{proof}
 Fix any such pair $i,j$.  Clearly, the claim is true for $e_{i,j} = e_q$, so consider any such edge $e_{i,j} \neq e_q$.
 For the first inequality (\ie the lower bound), we have the following:
 $$ r'_{{i,j}}\cdot n_{{i,j}} = r'_{{i,j}} \left\lceil \frac{r'_q}{r'_{{i,j}}} \cdot n_q\right\rceil \geq r'_{{i,j}} \cdot \frac{r'_q}{r'_{{i,j}}} \cdot n_q = r'_q \cdot n_q. $$
 For the second inequality (\ie the upper bound), we get:

 \begin{align*}
 r'_{{i,j}}\cdot n_{{i,j}} &= r'_{{i,j}} \left\lceil \frac{r'_q}{r'_{{i,j}}} \cdot n_q\right\rceil \leq r'_{{i,j}}\left(\frac{r'_q}{r_{{i,j}}} \cdot n_q + 1\right) = r'_q \cdot n_q + r'_{{i,j}} \leq r'_q \cdot n_q + r'_q\\
 &\leq r'_q \cdot n_q + r'_q \cdot \delta \left\lceil \frac{1}{\delta}\right\rceil = r'_q \cdot n_q + r'_q \cdot \delta \cdot n_q = (1+\delta)r'_q \cdot n_q.
 \end{align*}
 \end{proof}

We have thus defined an instance $I(G,s,k)$ of $k'$-median, consisting of $n+1$ candidates $C$, and $|E|n_qr'_q+\sum_{i,j} n_{{i,j}} r'_{i,j}$ clients $A$
(where the sum is taken over pairs $i,j$ such that $(v_i,v_j) \in E$). The following lemma shows that the $k'$-median instance $I(G,s,k)$ has a small cost if and only if
the initial graph has a small partial vertex cover.

\begin{lemma}
The graph $G$ has a partial vertex cover of size $k$ covering at least $s$ edges if and only if $I(G,s,k)$ has a $k'$-median solution of cost at most $\nu=\mu(1+\delta)(s + (m-s)(1+\varepsilon))$.
\end{lemma}

\begin{proof}
For the first direction, assume that $G$ has a partial vertex cover of size $k$ covering at least $s$ edges, and denote by $S$ the partial vertex cover solution.
Then for each vertex in the solution $v_i \in S$, we open a center at $\tilde{v}_i$, as well as one at $z^*$ (hence, we open $k'=k+1$ candidate centers in total).
Let $e_{i,j}=(v_i,v_j)$ be one of the $s$ edges that is covered in $G$, which corresponds in the reduction to $n_{i,j}$ client points placed at $c'_{{i,j}}$. By construction, $e_{i,j}$ is covered either by $v_i$ or $v_j$, and thus one center is opened either at $\tilde{v}_i$ or $\tilde{v}_j$. We have $d(c'_{i,j},\tilde{v}_i)=d(c'_{i,j},\tilde{v}_j)= r'_{i,j}$, and thus the cost induced by the client points placed at $c'_{i,j}$ is at most $cost_{i,j}$, which is at most $(1+\delta)\mu$ by Lemma~\ref{L:discretization}. On the other hand, for the edges $e_{{i,j}}$ that are not covered in $G$, the associated client points can be served by the candidate center at $z^*$, inducing a cost of $(1+\varepsilon)cost_{i,j}\leq(1+\delta)(1+\varepsilon)\mu$. Finally, the client points at $z^*$ have no cost since $z^*$ is also an open center. Thus the cost of the instance is bounded by $\nu$.

  For the other direction, assume that we have a $(k+1)$-median solution for $I(G,s,k)$ of cost at most~$\nu$. We first claim that this means that a center is opened at~$z^*$. Indeed, the closest other candidate center is at $\tilde{v}_1$, which is at distance at least $2$ from $z^*$. Serving all the client points located at $z^*$ would therefore cost at least $n_{z^*}2=2m\mu$, which is strictly larger than $\nu$ for a sufficiently small $\varepsilon>0$. Thus, there is a center open at $z^*$, which serves the clients located there, so in the rest of the proof we can ignore the cost of such clients.

  By construction, for each edge $e_{i,j}=(v_i,v_j)$, the two closest candidate centers to $c'_{i,j}$ are at distance $r'_{i,j}$, while all the other candidate centers are at distance at least $(1+\varepsilon)r'_{i,j}$. Thus, the cost of covering the $n_{i,j}$ client points located at $c'_{i,j}$ is  $cost_{i,j}$ if a center is opened at $\tilde{v}_i$ or $\tilde{v}_j$, and $(1+\varepsilon)cost_{i,j}$ otherwise (since in such a case it can be served by the center at $z^*$). We denote by $E_1$ the set of edges $e_{i,j}$ for which a candidate center is open
  at one of the nearby candidate locations $\tilde{v}_i$ or $\tilde{v}_j$.  Then the cost of the $(k+1)$-median instance is

  \[\sum_{e_{i,j} \in E_1} cost_{i,j} + \sum_{e_{i,j} \in E\setminus E_1} (1+\varepsilon)cost_{i,j}\geq \mu(|E_1|+ (m-|E_1|)(1+\varepsilon)),\]
where the inequality comes from Lemma~\ref{L:discretization}. By hypothesis, the cost is bounded by $\nu$, and for $\delta < \frac{\varepsilon}{m(1+\varepsilon)}$, this shows that $|E_1|\geq s$. Thus, we can cover at least $s$ edges of $G$ by taking the vertices $v_i$ for which a candidate center is opened at the corresponding candidate center $\tilde{v}_i$, and therefore $G$ has a partial vertex cover of size $k$ covering at least $s$ edges.
\end{proof}

It remains to show that the reduction takes time polynomial in $n$ and linear in $k$.

\begin{lemma}
Starting from a graph $G$ on $n$ nodes and a parameter $k$, we can compute the corresponding instance $I(G,s,k)$ of $k'$-median in time $k+poly(n)$.
\end{lemma}

\begin{proof}
Note that since the parameter $k$ is never used in the reduction (other than for determining $k'$), the only cost associated with it is essentially copying it from one instance to
the other, so the overhead is at most $k$ (actually it is much less). Then, placing the candidate centers $z^*$ and $\tilde{v}_i$ on the moment curve is straightforward since their coordinates are polynomials. However, placing the clients is a more delicate matter. We claim that all the computation associated with them only involves rational fractions of constant degree, and they can be carried out in polynomial time.

We first compute the points $c_{i,j}$, which are circumcenters of five points on the moment curve. This can be done in polynomial time since it amounts to computing the intersections of four bisector hyperplanes, and hence solving a linear system of constant size. Furthermore, since the points $\tilde{v}_i$ have integer coordinates, the solution of the system is a rational fraction. In particular, the squares of the circumradii are rational as well. Finally, in the perturbation scheme of Lemma~\ref{L:perturbation}, the radius $r$ of the ball in which we perturb can be taken to be a rational fraction of the input as well since the squares of the distances between $c_{i,j}$ and the $\tilde{v}_k$ are rational fractions. Therefore, one can also choose $\varepsilon$ and $c'_{i,j}$ to be rational fractions, and thus $\delta$ can be taken to be inverse polynomially bounded in the input. This bounds the size of the set of clients by a polynomial. Since all the variables in the cost $\nu$ of the $I(G,s,k)$ instance are rational fractions of the input, it can be computed in polynomial time as well, which concludes the proof.
\end{proof}

This concludes the proof of Theorem~\ref{thm:pc:4d}.

\section{Hardness of $k$-Median with Penalties in Two Dimensions}\label{sec:2dhardness}
In this section, we show that there is no algorithm running in time less
than $f(k)n^{o(\sqrt{k})}$ for any computable function $f$ that solves the $k$-median with
penalties problem in two dimensions (under the ETH assumption).
We do so by reduction from a problem called Grid Tiling introduced in~\cite{marx2007optimality}, which we now define.
\begin{Definition}[Grid Tiling]~\\
  \textbf{Input:} Integer $n$, collection $\mathcal{S}$ of $k^2$ nonempty sets $S_{i,j} \subseteq [n] \times [n]$ (where $1 \leq i,j \leq k$).\\
  \textbf{Parameter:} Integer $k$.\\
  \textbf{Output:}  \texttt{YES} if and only if there exists 
  a set of $k^2$ pairs $s_{i,j} \in S_{i,j}$ such that
	\begin{itemize}
		\item If $s_{i,j} = (a,b)$ and $s_{i+1,j} = (a',b')$, then $a = a'$.
		\item If $s_{i,j} = (a,b)$ and $s_{i,j+1} = (a',b')$, then $b = b'$.
	\end{itemize}
\end{Definition}
It is known that this problem has no $f(k)n^{o(k)}$-time algorithm unless ETH fails~\cite{CFKLMPPS15}.  In fact, 
we reduce from a slightly different version of the problem where, instead of equality, we have inequality
constraints of the following form:
\begin{itemize}
	\item If $s_{i,j} = (a,b)$ and $s_{i+1,j} = (a',b')$, then $a \leq a'$.
	\item If $s_{i,j} = (a,b)$ and $s_{i,j+1} = (a',b')$, then $b \leq b'$.
\end{itemize}
We call this problem Grid Tiling Inequality, and it is also known that this problem has no $f(k)n^{o(k)}$-time algorithm
unless ETH fails~\cite{CFKLMPPS15}.

Our reduction is similar in spirit to one given by Marx~\cite{marx2007optimality} for Independent Set of Unit Disks.
In the following, it is helpful to imagine the reduction in
a continuous setting in which the client points are infinite and uniformly placed in some
region (ultimately, we discretize this region so that we work with finitely many client points).
Note that the cost of a point is either the distance to its closest open center or its penalty, whichever is smaller.

\begin{theorem}
\label{thm:pc:2d}
There is no $f(k)n^{o(\sqrt{k})}$-time algorithm for the $k$-median with penalties problem in $d=2$ dimensions (for any computable
function $f$), unless ETH fails, where $n$ is the size of the input.
\end{theorem}
\begin{proof}
As mentioned, we reduce from the Grid Tiling Inequality problem.  For a fixed parameter $k$, we are given as input
an integer $n$ and a collection of sets $\mathcal{S}$ of $k^2$ nonempty sets $S_{i,j} \subseteq [n] \times [n]$ for all $1 \leq i,j \leq k$.
We show how to construct a $k'$-median with penalties instance that is a yes-instance if and only if the input to Grid Tiling Inequality
is a yes-instance. We set $k'=k^2$, which shows the claimed lower bound: suppose towards a contradiction there exists an
algorithm running in time $f(k')n^{o(\sqrt{k'})}$ for the $k'$-median with penalties problem (where $f$ is some computable function).  Then this means there is an algorithm
running in time $f(k^2)n^{o(\sqrt{k^2})} = f(k^2)n^{o(k)}$ that solves the Grid Tiling Inequality problem (with parameter $k$), yielding a contradiction under
the ETH assumption.

The instance for the $k'$-median with penalties problem is as follows.  As mentioned, we have $k' = k^2$, and we fix $\epsilon = 1/n^3$. The client points lie in the region consisting of a square of side length $2k + \epsilon(n-1)$, where the lower left corner of the square is on the origin (\ie $A = \{(x,y) \mid 0 \leq x,y \leq 2k + \epsilon(n-1) \}$). They are spaced evenly in a grid $G$, where two consecutive (horizontal or vertical) clients are at a distance $\varepsilon$ from each other, and thus there are $\Sigma=(2k/\varepsilon+n)^2$ clients. Each client point $a$ has a penalty of $p_a = 1$. We think of this grid as a discrete approximation of the uniform measure on the square $A$, and in line with this analogy, we work with the discrete measure $\mu$ carried by the client points, where each client is weighted $1$, so that $\iint_A d\mu=\Sigma$.

For each set $S_{i,j}$, we introduce $|S_{i,j}| \leq n^2$ candidate centers, and we let $C_{i,j}$ denote the set of such candidate
centers (note that there are $k^2$ such sets), where $C_{i,j} = \{(2i - 1,2j-1) + \epsilon(u-1,v-1) \mid (u,v) \in S_{i,j}\}$. Note that the candidate centers are also placed on vertices of $G$, and that, if $S_{i,j}$ has all possible pairs so that $S_{i,j} = [n] \times [n]$,
then $C_{i,j}$ precisely forms a subgrid of $n^2$ evenly spaced points in which consecutive points are at distance $\epsilon$ from each other and
the lower left point of the subgrid lies at $(2i-1,2j-1)$. The final set of candidates is given by $C = \cup_{1 \leq i,j \leq k} C_{i,j}$.
For now, we defer defining the cost threshold $\nu$.

Note that, when opening a candidate center, it can only serve
client points that are within a distance of $1$ to it, since all other client points $a$ would rather pay the penalty $p_a = 1$.  Moreover, for each
candidate center $c \in C$, we have the property that the entire disk $D$ of radius $1$ centered at $c$ is completely contained in the square region $A$.
Indeed, consider any candidate center $c_{i,j}$ corresponding to the pair $(u,v) \in S_{i,j}$ (for $1 \leq i,j \leq k$, $1 \leq u,v \leq n$),
so that $c_{i,j} = (2i-1,2j-1) + \epsilon(u-1,v-1)$.  The leftmost point possible is given by $u = 1, i = 1$, which yields a point of the form $(1,2j-1 + \epsilon(v-1))$.
Hence, no disk of radius $1$ centered at a candidate center goes beyond the left edge of the square~$A$.  The rightmost possible point is given by $i = k,u=n$, which yields a point
of the form $(2k-1 + \epsilon(n-1), 2j-1 + \epsilon(v-1))$.  Hence, no disk of radius $1$ centered at a candidate center goes beyond the right edge of the square
(which lies at $x = 2k + \epsilon(n-1)$).  A similar argument shows that no such disk goes beyond the upper or lower edges of $A$.

We now seek to understand the costs of solutions in which such disks intersect, and compare them to solutions in which they do not intersect. Consider a collection $\Delta$ of $k^2$ pairwise disjoint disks centered on candidate centers
(with the possible exception that pairs may intersect at exactly one point on the boundary, so that they are tangent to each other),
each of which has radius $1$ and is fully contained in $A$.  We claim that any such solution (obtained by opening a candidate at the center of each disk)
has the same cost. To see this, note that each candidate center $c$ contributes the same amount to the cost of the solution: this contribution
is given by the double integral $\iint_{D} d((x,y),c)d\mu$, where $D$ denotes the disk of radius $1$ centered at $c$. Since the candidate centers are placed on vertices of $G$, any candidate center sees exactly the same configuration of clients in $D$, and thus this integral does not depend on~$c$.

Now, all other points that do not belong to one of these $k^2$ disks pay a penalty of $1$.  Hence, such points contribute $\iint_{A\setminus \Delta} d\mu$ to the cost, since $A \setminus \Delta$ is the region of the square $A$ that they occupy. Similarly as before, since the disks are pairwise disjoint and the candidate centers are placed on vertices of $G$, this quantity does not depend on the actual placement of the centers. In total, the cost of a solution where the candidate centers induce a family of disjoint disks $\Delta$ is $k^2\iint_{D} d((x,y),c)d\mu+\iint_{A\setminus \Delta} d\mu$. This quantity does not depend on the placement of the center, and we set the cost threshold $\nu$ of our instance to be this value. Note that since $\mu$ is a discrete measure, the integrals are actually sums, and thus $\nu$ can trivially be computed in polynomial time from the input instance of Grid Tiling Inequality.  Actually, the value of $\nu$ is in general irrational, and hence we need to slightly increase it to make it rational.  We briefly discuss how to deal with this issue after the following discussion of costs of solutions where disks intersect.

On the other hand, consider a solution $S_1$ in which at least one pair of the $k^2$ disks intersect each other. We look at the Voronoi diagram induced by the centers opened in this solution, and denote by $V_i$ the Voronoi region corresponding to a center $c_i$. As before, a center only serves points at distance at most $1$ from it, so it serves points in the region $R_i:=V_i \cap D(c_i,1)$,
where $D(c_i,1)$ denotes the disk of radius $1$ centered at $c_i$. The total cost of the solution $S_1$ is thus $\sum_i \iint_{R_i}d((x,y),c_i)d\mu+ \iint_{A\setminus \cup_i R_i} d\mu$. Since for all $i$, $R_i\subseteq D(c_i,1)$, at least one of these inclusions is strict, and the distance $d((x,y),c_i)$ in the integrals is always strictly less than $1$, we have: 
\begin{align*} \sum_i \iint\limits_{R_i}d((x,y),c_i)d\mu+ \iint\limits_{A\setminus \cup_i R_i} d\mu &= \sum_i \iint\limits_{R_i}(d((x,y),c_i)-1)d\mu + \iint\limits_A d\mu \\
 & > \sum_i \iint\limits_{D(c_i,1)}(d((x,y),c_i)-1)d\mu + \iint\limits_A d\mu\\
  &= \sum_i \iint\limits_{D(c_i,1)}d((x,y),c_i)d\mu+ \iint\limits_Ad\mu-\sum_i\iint\limits_{D(c_i,1)}d\mu=\nu.
\end{align*}
This proves that a solution where disks of radius $1$ (centered at the opened candidate centers) intersect always has a cost strictly greater than the cost threshold~$\nu$.
Regarding irrationality of $\nu$, note that the minimum possible cost of a solution in which disks intersect is attained when all disks centered at candidates are disjoint,
with the exception of one pair that intersect (in the smallest area possible).  Since candidates are placed on a grid, the smallest such intersection that can be obtained is when two candidate
centers are opened that are at distance $2-\Omega(\epsilon)$ from one another.  Hence, we can choose $\nu$ to be a rational number in between the cost of such a solution and
$k^2\iint_{D} d((x,y),c)d\mu+\iint_{A\setminus \Delta} d\mu$.

To finish the theorem, we need only argue that there is a solution to the Grid Tiling Inequality input if and only if it is possible to select $k^2$
pairwise disjoint disks of radius $1$, where each disk is centered at some candidate. This is proved in the aforementioned reduction of Marx~\cite{marx2007optimality}, but we include it here for completeness. Note that we allow intersection
at exactly one point.  To this end, suppose we are given a yes-instance for the Grid Tiling Inequality problem.  For each $1 \leq i,j \leq k$, let $s_{i,j} = (u,v)$
denote the chosen pair in $S_{i,j}$, and open a candidate center $c_{i,j}$ at the point $(2i-1,2j-1) + \epsilon(u-1,v-1)$.  Now, the only possible disks that can
intersect with the disk of radius $1$ centered at $c_{i,j}$, denoted by $D_{i,j}$, are $D_{i+1,j},D_{i-1,j},D_{i,j+1}$, and $D_{i,j-1}$ (if such disks exist).
This holds since all other candidates have distance at least $\sqrt{2(2-\epsilon(n-1))^2} = \sqrt{2}(2 - \epsilon(n-1))$ to $c_{i,j}$, which is at least $2$ for sufficiently
small $\epsilon$.  We only argue that the disks $D_{i,j}$ and $D_{i+1,j}$ do not intersect, since the other cases follow by a similar argument.
In particular, let $s_{i,j} = (u,v)$ and $s_{i+1,j} = (u',v')$, and note that $u \leq u'$ (since the input is a yes-instance).  Hence,
the distance between $c_{i,j}$ and $c_{i+1,j}$ is given by:
\begin{align*}
&\sqrt{(2i+1 + \epsilon (u'-1) - (2i - 1 + \epsilon (u-1)))^2 + (2j-1 + \epsilon (v'-1) - (2j - 1 + \epsilon (v-1)))^2}\\
& \qquad \qquad \qquad \geq \sqrt{(2 + \epsilon(u' - u))^2} \geq 2,
\end{align*}
and hence the disks do not intersect (since both of them have a radius of $1$).

Now suppose we have a yes-instance for the $k'$-median with penalties problem.  We seek to show that we have a yes-instance for the Grid Tiling Inequality problem.
In particular, since the cost is at most $\nu$, we know that there is a way of selecting $k^2$ candidate centers $c_{i,j}$ (for $1 \leq i,j \leq k$)
where their corresponding disks $D_{i,j}$ of radius $1$ are pairwise disjoint.  This implies that, from each set $C_{i,j}$,
we have selected exactly one candidate center which is of the form $c_{i,j} = (2i-1,2j-1) + \epsilon(u-1,v-1)$ for some $(u,v) \in S_{i,j}$.
We claim that, for each $S_{i,j}$, taking such a pair $(u,v)$ satisfies the conditions of the Grid Tiling Inequality problem.  In particular,
consider any $(u,v) \in S_{i,j}$ and $(u',v') \in S_{i+1,j}$.  We want to show that $u \leq u'$.  Since the disks $D_{i,j}$ and $D_{i+1,j}$
do not intersect, the distance between them is at least $2$, which means:
\begin{align*}
2 &\leq \sqrt{(2i + 1 + \epsilon(u'-1) - (2i - 1 + \epsilon(u-1) ))^2 + (2j-1 + \epsilon(v'-1) - (2j - 1 + \epsilon(v-1) ))^2} \\
&=	\sqrt{(2 + \epsilon(u' - u))^2 + (\epsilon(v'-v))^2} \leq \sqrt{4 + 4\epsilon(u' - u) + 2\epsilon^2(n-1)^2}.
\end{align*}
Squaring both sides, we see that
\[ 4 \leq 4 + 4\epsilon(u' - u) + 2\epsilon^2(n-1)^2 \Longleftrightarrow -2\epsilon^2(n-1)^2 \leq 4\epsilon(u' -u) \Longleftrightarrow u - \frac{\epsilon(n-1)^2}{2} \leq u'.\]
As long as $\frac{\epsilon(n-1)^2}{2} < 1$, then we know that $u' > u - 1$.  Since $u'$ is an integer, we must have $u' \geq u$.  This holds
for a sufficiently small $\epsilon$ (\eg $\epsilon < \frac{2}{(n-1)^2}$).  The case regarding $s_{i,j} = (u,v) \in S_{i,j}$ and $s_{i,j+1}=(u',v') \in S_{i,j+1}$ implying
$v \leq v'$ is symmetric, and hence the proof is complete.

\end{proof}

\section{An Algorithm for $k$-Median in Two Dimensions}\label{sec:2d}
We define an instance of the 2-dimensional $k$-median optimization problem to be a triple $(C,A,k)$ where $C$ denotes the
set of candidates and $A$ denotes the set of clients. The output is a set $K$ of $k$ candidates, which has a cost of $\sum_{a\in A}d(a,K)$,
such that no other set of $k$ candidates obtains a lower cost. (See Definition~\ref{def:kmed} for the decision version of the problem.) We show:

\begin{theorem}
  \label{thm:upperbound:2d}
  There exists an exact algorithm that finds an optimal solution 
  to any instance $(C,A,k)$ of the 2-dimensional $k$-median optimization  problem in time $|A| \cdot |C|^{O(\sqrt{k})}$. 
\end{theorem}

Our algorithm is quite standard as it uses similar ideas 
as the ones in the work of Marx and Pilipczuk~\cite{MarxP15}. Since it
consists of guessing the set of candidate centers and their Voronoi cells, it
also works verbatim for $k$-means, as well as for the versions of $k$-means and $k$-median with penalties.

Let $(C,A,k)$ be an instance of the 2-dimensional $k$-median
problem. By a small perturbation of the positions of the candidate
centers $C$, we can assume that no point in $\mathbb{R}^2$ is
equidistant to four or more centers. Indeed, a small enough
perturbation will not change which centers are opened in an optimal
solution, and will slightly change the cost, but this cost can be
recomputed afterwards with the exact positions of the centers. The set
of points that are equidistant from $3$ candidate centers is denoted
by $P$, and there are $O(|C|^3)$ of them.

We define a \emph{separating curve} $S$ with respect to $C$ of length $r$ to be a concatenation of segments of the form  $(c_1, p_1), (p_1,c_2), \ldots , (c_r, p_r), (p_r,c_1)$, where the $c_i$ are candidate centers, and the $p_i$ are points in $P$ (see Figure~\ref{fig:voronoi}). A separating curve is \emph{valid} if it is simple, \ie there are no self-intersections.  We denote by $\ins(S)$ and $\out(S)$ its respective interior and exterior.

Our algorithm (Algorithm~\ref{algo:Clustering}) works by enumerating valid separating curves of size $O(\sqrt k)$, using them to cut the instance into two subinstances and recursing. The base cases are then solved by brute-force. The rationale behind this, which we formalize in the next subsection, is that since the Voronoi diagram of the optimal solution is a planar graph, it admits small balanced separators, which in this case can be realized by valid separating curves. Therefore, one of the separating curves we enumerate corresponds to such a small balanced separator, and as we will prove, such a separator can be easily used to partition the problem into two independent subinstances.

\renewcommand{\algorithmicforall}{\textbf{for each}}

\begin{algorithm}
  \caption{\texttt{ExactClustering}:
    Exact Algorithm for 2-Dimensional $k$-Median}
  \label{algo:Clustering}
  \begin{algorithmic}[1]
    \State \textbf{Input:} A set of candidate centers $C$,
    a set of clients $A$, a positive integer 
    $k$, and a set of centers that are already open $\hat{C}$
        \If{$k = O(1)$}
    \State Let $\calS = \text{argmin}_{S' \subseteq C, |S'| \le k} 
    \sum_{a \in A} \dist(a, S' \cup \hat{C})$
    \State Return $\calS \cup \hat{C}, \cost(\calS \cup \hat{C})$
    \EndIf
    \State $\calS \gets$ an arbitrary solution
    \State $P \gets$ set of points equidistant to three candidate centers
    \ForAll{valid separating curve $\gamma=(c_1,p_1), \ldots, (c_\ell,p_\ell),(p_\ell,c_1)$ of length 
      $\ell \le  \sqrt{4.5k}$}
      \State $\calS' \gets$ an arbitrary solution
    \State $\tilde{C} \gets \{c_1,\ldots,c_{\ell}\}$ 
    \ForAll{$k' \in \{k/3-\ell,\ldots,2k/3\}$}
    \State $\Sin, \cost(\Sin) \gets $ \texttt{ExactClustering}
    ($C\setminus \tilde{C} \cap (\ins(\gamma) \cup \gamma)$, 
    $A \cap (\ins(\gamma) \cup \gamma)$, $k'$, $\hat{C} \cup  \tilde{C}$)
    \State $\Sout,\cost(\Sout) \gets $ \texttt{ExactClustering}
    ($C \setminus \tilde{C} \cap (\out(\gamma) \cup \gamma)$, $A \cap \out(\gamma)$, $k-\ell-k'$, $\hat{C} \cup  \tilde{C}$)    
    \If{$\cost(S') > \cost(\Sin) + \cost(\Sout)$}
    \State $\calS' \gets \Sout \cup \Sin$, 
    $\cost(\calS') = \cost(\Sin) + \cost(\Sout)$
    \EndIf
    \EndFor
    \If{$\cost(\calS) > \cost(\calS')$}
    \State $\calS \gets \calS'$,
    $\cost(\calS) = \cost(\calS')$
    \EndIf
    \EndFor
    \State Return $\calS, \cost(\calS)$
  \end{algorithmic}
\end{algorithm}

\subsection{Correctness}
We first recall some notions from topological graph theory that we rely on. 

For a plane graph $G$, a \emph{noose} of $G$ is a Jordan curve that intersects $G$ only at its vertices, and visits each of its faces at most once. The length of a noose is the number of vertices that it intersects. A noose $\gamma$ is $\alpha$-\emph{face-balanced} if the number of faces that are strictly enclosed or strictly excluded by $\gamma$ are both not larger than $\alpha |F(G)|$,
where $F(G)$ denotes the set of faces of~$G$. The following theorem of Marx and Pilipczuk~\cite[Corollary~4.17]{MP15} (see also~\cite{MarxP15}) shows that there exist nooses that form balanced separators of small size for the faces of $G$.

\begin{theorem}\label{T:noose}
Let $G$ be a connected $3$-regular graph with $n$ vertices, $m \geq 6$ edges embedded on a sphere. Then there exists a $2/3$-face-balanced noose for $G$ that has length at most $\sqrt{4.5n}$.
\end{theorem}

We apply this theorem to the Voronoi diagram $\calV$ induced by an optimal solution to the instance $(C,A,k)$. By our assumption, no point is equidistant to four centers, and therefore this graph is $3$-regular. However, $\calV$ is not a graph embedded on the sphere, and not even a plane graph, since it has infinite rays. We remedy this by adding three dummy centers to $\calV$ that are very far from the rest of the candidate centers and clients, and make it so that there are only $3$ rays going to infinity. Then we compactify the picture by embedding this new graph into a sphere with an additional point at the intersection of the three rays. The reader can verify that applying Theorem~\ref{T:noose} to this graph yields a noose of the original graph with the same guarantees (for the obvious extension of the definitions for graphs with infinite rays).

Since the edges of $\calV$ are straight lines, its faces are convex, and there is a candidate center in the interior of each face, this particular case also allows for a discrete description of nooses: each noose can be discretized by replacing a maximal subarc in a face by two straight-line segments between its endpoints and the candidate center in that face. By Theorem~\ref{T:noose} applied to $\calV$ and this discretization, we obtain that there exists a valid separating curve $\gamma$ of length at most $\sqrt{4.5k}$ that is a $2/3$-face-balanced noose for $\calV$.

The following lemma shows that such a separating curve can be used to partition the clients.  In the following, we denote by $\opt$ an optimal
solution.

\begin{lemma}
  \label{lemma:upperbound:correct}
  We have that all the clients 
  in $\ins(\gamma)$ (respectively $\out(\gamma)$) are served by a 
  center in $\opt$ that is
  in $\ins(\gamma) \cup \gamma$ (respectively $\out(\gamma) \cup \gamma$).
\end{lemma}
\begin{proof}
Assume towards a contradiction that there is a client $a$ that is 
  in $\ins(\gamma)$ and is served in $\opt$ by a center $c$
  in $\out(\gamma)$. Consider the line $L$ from $a$ to $c$. Since $\gamma$ is a Jordan curve, $L$ has to intersect at least one of the line segments defining $\gamma$, which links a center $c'$ to a point $p$. The intersection point lies in the Voronoi cell associated with $c'$, and thus by the triangle inequality $a$ is closer to $c'$ than $c$, a contradiction.
\end{proof}

We can thus finish the proof that Algorithm~\ref{algo:Clustering} is 
correct. Among all the sequences of pairs $(c_i,p_i) \in C \times P$, 
one induces $\gamma$. Therefore, an immediate induction on the recursive
calls yields the result.

\subsection{Running Time}
\begin{lemma}
  \label{lemma:upperbound:RT}
  The running time of Algorithm~\ref{algo:Clustering} on
  a 2-dimensional instance $(C,A,k)$ of the $k$-median problem 
  is at most $|A|(|C|)^{O(\sqrt{k})}$.
\end{lemma}
\begin{proof}
Recall that we are working in a computational model where sums of square roots can be compared efficiently, which allows us to compare sums of distances, even though they might be irrational.
  
  Observe first that the cost of a solution can be evaluated in time
  $O(|A| \cdot k)$. Hence, the final recursive call takes time 
  at most $|A| \cdot |C|^{O(1)}$.

  Since $P$ has size $O(\left|C\right|^3)$, there are only $O((\left|C\right|^{3})^{2\sqrt{k}})$ choices of a valid separating curve,
which is at most $\left|C\right|^{O(\sqrt{k})}$.

  We now consider the general recursive calls, which yield the following recurrence:
  $T_k \le |C|^{O(\sqrt{k})} 2k T_{2k/3}$ (where $T_k$ denotes the running time of the algorithm when given $k$ as input). Hence,
  the running time is at most $|A||C|^{O(\sqrt{k})}$.
\end{proof}

\begin{figure}[ht!]
\begin{center}
    \begin{tikzpicture*}{0.3\textwidth}
        \begin{scope}[
            vertex/.style={
                draw,
                circle,
                minimum size=2mm,
                inner sep=0pt,
                outer sep=0pt%
            },
            tinyvertex/.style={
                draw,
                circle,
                minimum size=1mm,
                inner sep=0pt,
                outer sep=0pt%
            },
            every label/.append style={
                rectangle,
                %label distance=.1cm,
                font=\small,
            }
            ]

\node[vertex,label={above:$c_i$}] (c1) at (0,2) {};
\node[vertex,label={above:$c_{i+1}$}] (c2) at (10,2) {};
\node[vertex,label={below right:$p_i$}] (p1) at (5,1) {};
\node[tinyvertex] (c3) at (5,4) {};

 \draw plot coordinates {(0,2) (5,1) (10,2)}; 
  
 \draw plot coordinates {(5,0) (5,1)}; 
 \draw plot coordinates {(3,3) (5,1) (7,3)}; 
 \draw plot coordinates {(3.7,-1) (5,0) (6.3,-1)}; 
\end{scope}
\end{tikzpicture*}
\end{center}
    %    \vspace{-0.25cm}
    \caption{A segment of a valid separating curve: $c_i - p_i - c_{i+1}$.}
\label{fig:voronoi}
\end{figure}
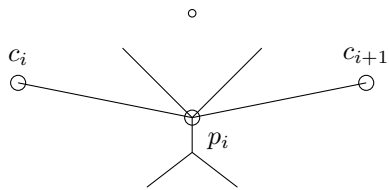

\subparagraph*{Acknowledgments} We are grateful to Dominique Attali for very helpful discussions.

\bibliographystyle{siam}
\bibliography{hardness}

\end{document}